\documentclass[pra,twocolumn,showpacs]{revtex4-2}

\usepackage{hyperref,mathtools,graphicx,cases}
\usepackage{amssymb,amsthm,amsfonts,amsmath}
\usepackage{physics}
\usepackage{color}

\newcommand{\M}{\mathsf{M}}
 
\newcommand{\T}{\mathsf{T}}

\newcommand{\ef}{\mathcal{E}(\mathcal{H})}

\newcommand{\lh}{ \mathcal{L}(\mathcal{H})}
\newcommand{\lk}{ \mathcal{L}(\mathcal{K})}
\newcommand{\hk}{\mathcal{K}}
\newcommand{\h}{\mathcal{H}}							

\newtheorem{theorem}{Theorem}
\newtheorem{proposition}{Proposition}

\begin{document}
\title{Single-shot labeling of quantum observables}
\author{Nidhin Sudarsanan Ragini}
\author{Mario Ziman}
\affiliation{RCQI, Institute of Physics, Slovak Academy of Sciences, D\'ubravsk\'a cesta 9, 84511 Bratislava, Slovakia}
\begin{abstract}
We identify and study a particular class of distinguishability problems for quantum observables (POVMs), in which observables  with permuted effects are involved, which we call as the\textit{ labeling problem}. Consequently, we identify the binary observables those can be “labeled” perfectly. In this work, we study these problems in the single-shot regime. 
\end{abstract}
		
\maketitle
\section{Introduction}
\noindent The lack of understanding of the individual action performed by a measurement apparatus is the essence of most of quantum controversies. Quantum theory provides a way to describe its probabilistic aspects for all practical purposes. In particular, when one is interested to predict the statistics of observed outcomes, each measurement outcome is fully determined by a positive operator known as effect (see for example \cite{heinosaari2011mathematical}). An effect itself does not explain the physics, or form of the recorded event, but captures the probabilities of its occurrence. In order to describe measurement statistics, we assign labels and effects for all outcomes. While effects associated with the measurement apparatus can be verified experimentally, the choice of labels cannot be tested by the measurement apparatus itself.  
        
In this paper, we consider the situation when the mathematical description of effects is known; However, their connection to the labels are lost. The question of outcomes' labeling we are going to address is the task of assigning a mathematical description to the unlabeled outcomes, assuming the effects forming the observable are known. For example, consider a black box with green and red diodes. Each time the measurement is accomplished one of them shines and outcome is recorded. We know the meaurement apparatus is designed to decide whether the measured atom is in its excited, or its ground state. It is not perfect and the associated effects are $E=\eta\dyad{e}$ and $G=(1-\eta)\dyad{e}+\dyad{g}$, respectively. The parameter $\eta$ quantifies the imperfections in registration of excited state. Now, is it possible to identify how the colors and effects are paired? How many measurement runs are needed? It is straightforward to realize that measuring the atom prepared in the ground state can lead only to outcome described by the effect $G$. Therefore, the diode that shines in this situation is necessarily described by the effect $G$. The task of outcome labeling is accomplished and the labels (shining diodes) are identified with effects.  

It turns out that this question is a special case of the discrimination problem whose origin stems from the seminal work of Helstrom \cite{helstrom1969quantum}, who investigated the discrimination problem for quantum states. In particular, in its simplest form we are given a black box device known to be either A or B. Our task is to design an experiment enabling us to conclude whether A or B in a single run of the experiment. The conclusions might be of various relevance or significance. One of the measure is their averaged error probability, known to justify operationally the choice of trace norm as the measure of statistical distance between A and B \cite{helstrom1969quantum}. The outcome labeling problem of $n$-valued measurement aims to discriminate among $n!$ permutations of different labelings, as such, among $n!$ devices. Unfortunately, very little is known about decision problems beyond the binary case. 

The performance of measurements is crucial for the successful development of quantum technologies (see e.g. Ref.~\cite{HRT2023}). Therefore, the elementary, as well as, foundational question of discrimination of quantum measurements has attracted researchers and various versions and aspects of the problem has been studied in Refs.\cite{JFDY2006,ZH2008,ZHS2009,SZ2014,PPKK2018,PPKK2021}. In all of them, either the labeling is assumed to be known, or the equivalence classes ignoring the labelings are considered as descriptions of unknown (and unlabeled) measurement devices. The question of assigning unknown labels was not addressed explicitly. The closest is the work in Ref.~\cite{krawiec2020discrimination}, where authors analyzed the minimum-error discrimination of two SIC POVMs - one of them being an arbitrary, but fixed permutation of the other. The question of labeling is about identification of the permutation, however, in the mentioned paper the permutation is fixed, thus, the problem is a binary decision problem that can be investigated in single shot settings. However, the labeling problem for SIC POVM aims to distinguish among $d^2!$ different permutations ($d$ is the Hilbert space dimension of the measured system).

The paper is organized as follows. In Section II the general problem of labeling is formulated. The Section III is focused on perfect labeling for single shot settings. Minimum error labeling and unambiguous labeling are studied in Sections IV and Section V, respectively. The Section VI introduces the concept of antilabeling and the main results are summarized in Section VII. 

\section{Formulation of labeling problem}
\subsection{Mathematical tools}
\noindent The mutually exclusive outcomes \textcolor{black}{$x_1,\dots,x_n$} of $n$-valued quantum measurements are represented by \textit{effects},
i.e. positive operators $M_1,\dots,M_n$, respectively, satisfying
the normalization $M_1+\cdots+M_n=I$.
For most of the situations we can safely say the measurements
are collection of effects, however, rigorously, the measurements
are mappings assigning effects for particular outcomes.
Let us denote by $\Omega$ the ordered set of outcome \textit{labels}
$\Omega := \{x_1,\cdots,x_n\}$ and by $\ef$ the set of all effects.
Then, a quantum measurement is described by an \textit{observable}, as a normalised positive operator-valued
measure (POVM) $\M: \Omega\to\ef$ with $\M(x_j)=M_j$.
Alternatively,
we can think of observables as measure-and-prepare channels $\mathcal{M}$
transforming input states of the measured system
$\varrho$ into output states (effectively probability distributions)
of $n$-dimensional systems
\begin{equation}
\mathcal{M}(\rho) = \sum_{x_k \in \Omega}\text{Tr}[\rho M_k]\dyad{x_k}\,.
\end{equation}	 
Once we have the channel form, the Choi-Jamiolkowski isomorphism \cite{choi1975completely,jamiolkowski1972linear} between completely positive maps and positive operators associates the observable $\M$ with a positive operator
\begin{equation}
  \label{eq:cj_measure_prepare_form}
\mathfrak{M} = (\mathcal{I} \otimes \mathcal{M} )[\Psi_+] = \sum_{k} M_k^\top \otimes\dyad{x_k} \,,
\end{equation}
where $\Psi_+$ is the (unnormalized) maximally entangled state on $\h\otimes\h$ defined by $\Psi_+ := \dyad{\psi_+}$ with $\ket{\psi_+} := \sum_{j=1}^{d}\ket{j}\otimes\ket{j}$ and $\ket{1},\dots,\ket{d}$ form an orthonormal basis of $\h$.

\textcolor{black}{
  The mathematical framework describing the measurements of processes was introduced and developed in Refs.~\cite{gutoski2007toward,chiribella2008circuit,ziman2008process}. The most general single-shot experiment measuring the properties of quantum processes (channels) is illustrated in Fig.~\ref{fig2}. In particular, the channel under consideration $\mathcal{M}$ acts on the subsystem of a bipartite \textit{test state} $\varrho$ and, subsequently, the output is measured by a bipartite \textit{test observable} $\mathsf{F}$ leading to outcome $\alpha$. Suppose the input subsystem of the channel is associated with the Hilbert space $\h$ and its output is associated with Hilbert space $\hk$, thus, the channel $\mathcal{M}$ transforms operators $\lh$ into $\lk$. Further, let us denote by $\h_{\rm anc}$ the Hilbert space associated with the anciliary subsystem, thus, the test state $\varrho$ is a density operator from $\mathcal{L}(\h_{\rm anc}\otimes\h)$ and $F_\alpha\in\mathcal{L}(\h_{\rm anc}\otimes\hk)$ is the effect associated with the outcome $\alpha$ of the test observable $\mathsf{F}$. The outcome $\alpha$ occurs with probability $p_\alpha=\text{tr}\left[(\mathcal{I}\otimes\mathcal{M})[\varrho]F_\alpha\right]$.
  }

\begin{figure}[h]
\centering
\includegraphics[width=0.25\textwidth]{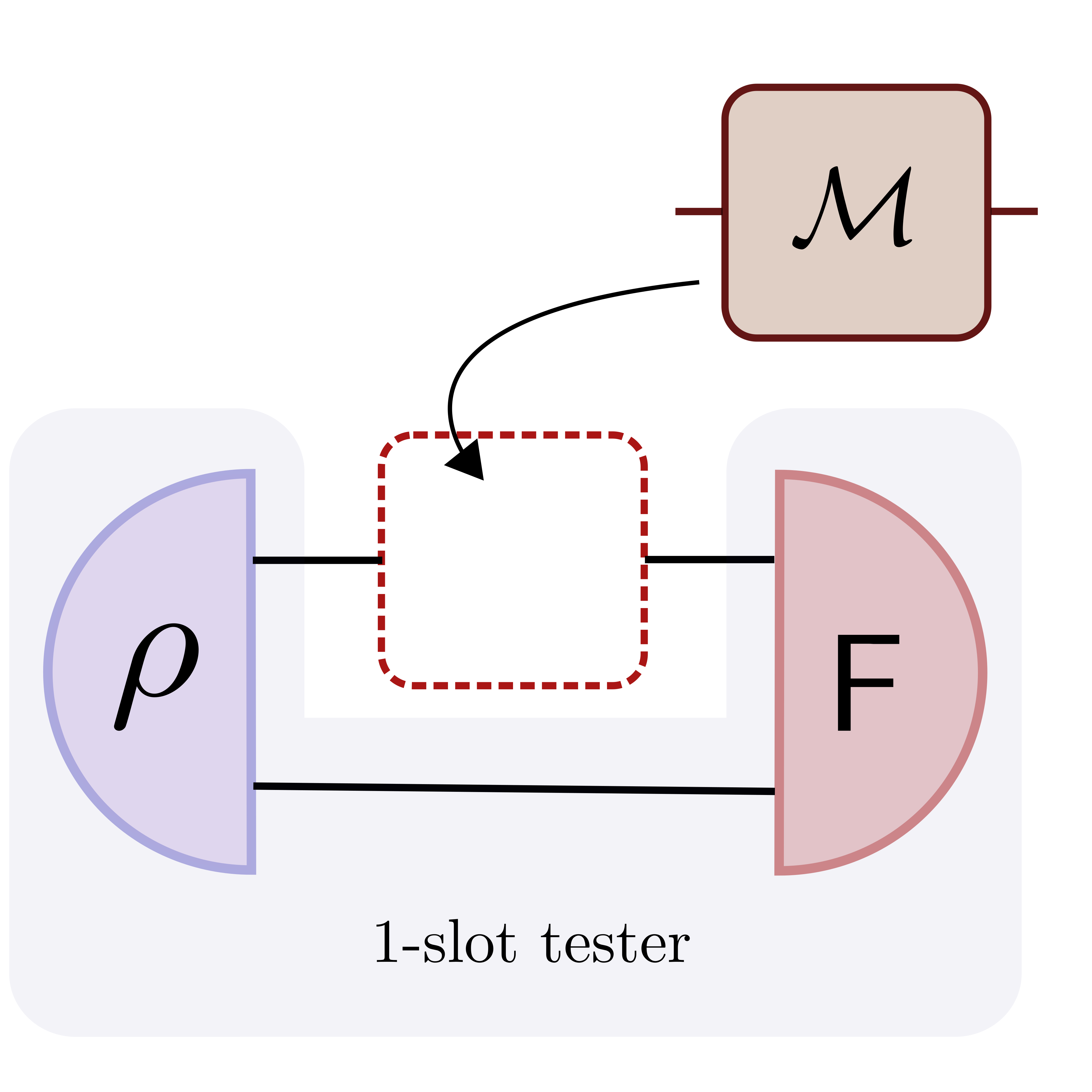}
\caption{\textcolor{black}{To measure a channel $\mathcal{M}$, we send a test state $\rho$ through an extended channel and perform a test observable $\mathsf{F}$ on the evolved state. We can identify that the pair of state and observable constitute the measurement procedure event and, as such, determines a \textit{quantum tester} $\mathsf{T}$. The reduced subsystems entering the unknown chanel is in the state $\xi^\top$, where the transposition is used just to obtain simpler expressions in the other formulas.}}\label{fig2}
\end{figure}

\textcolor{black}{
For any state $\varrho$ the Choi-Jamiolkowski isomorphism defines a completely positive linear map $\mathcal{R}_\varrho:\mathcal{L}(\h\otimes\h)\to\mathcal{L}(\h_{\rm anc}\otimes\h)$ such that $\varrho=(\mathcal{R}_\varrho\otimes\mathcal{I})[\Psi_+]$ (see Ref.~\cite{ziman2008process}). Using its dual $\mathcal{R}^*_{\rho}$ we obtain $p_\alpha=\text{tr}\left[(\mathcal{I}\otimes\mathcal{M})(\mathcal{R}_\varrho\otimes\mathcal{I})[\Psi_+]F_\alpha\right]=\text{tr}\left[(\mathcal{I}\otimes\mathcal{M})[\Psi_+](\mathcal{R}^*_\varrho\otimes\mathcal{I})[F_\alpha]\right]$. In the last expression we may identify the Choi-Jamiolkowski representation $\mathfrak{M}$ of the tested channel $\mathcal{M}$ and introduce the so-called \emph{process effect} $\mathsf{T}_\alpha=(\mathcal{R}^*_\varrho\otimes\mathcal{I})[F_\alpha]$ determining the probability of outcome $\alpha$ for any channel $\mathcal{M}$. As a result we obtain the description of the process measurement by so-called \emph{tester} being a collection of positive operators $\mathsf{T}_1,\dots,\mathsf{T}_m$ ($\mathsf{T}_\alpha\in\mathcal{L}(\h\otimes\hk)$) satisfying the normalization $\T_1+\cdots+\T_m=\xi\otimes I$, where $\xi=(\text{tr}_{\mathcal{K}}[\rho])^\top$ is the transposition of the reduced state of the subsystem passing through the unknown channel $\mathcal{M}$. We recover Born-like formula $p_{\alpha}={\rm tr}[\mathfrak{M}\T_\alpha]$ for the outcome probabilities.
}

\textcolor{black}{
Let us stress that for testing channels of measure-and-prepare form (Eq.~\ref{eq:cj_measure_prepare_form}) the tester operators can be assumed, without loss of generality, to be of the form  $\T_{\alpha} = \sum_{k} \mathsf{H}_{k}^{(\alpha)}\otimes \dyad{x_k}$ \cite{SZ2014}, where for all $k$, the  normalization $\sum_{\alpha}\mathsf{H}_{k}^{(\alpha)} = \xi$ is satisfied and all $\mathsf{H}_k^{(\alpha)}\in\mathcal{L}(\h)$ are positive operators.
  }
		
\subsection{Labeling problem}
\noindent The question of labeling takes place when the information on the mapping $x_k\mapsto M_k$ is lost, but the collection of effects is known. If all the effects are different, then for $n$-valued observable there are $n!$ different permutations how the outcomes can be paired with the effects. Therefore, the goal of labeling problem is to discriminate among $n!$ permutations. Let us fix the order of effects and denote by $\M_\sigma$ the observable obtained by $\sigma$ permutation. The result of labeling is the identification of $\sigma$. In what follows we will assume that all the permutations are equally likely, thus the apriori information on $\sigma$ is represented by probability distribution $\pi(\sigma)=1/n!$.

\textcolor{black}{In the language of measure-and-prepare channels the labeling takes the form of discrimination among channels, whose actions are given by $\mathcal{M}_\sigma(A)=\sum_k {\rm tr}[A\M_k]\dyad{x_{\sigma(k)}}=P_\sigma\mathcal{M}_{\rm id}(A)P_\sigma^\dagger$, where $P_\sigma =\sum_k \ket{x_{\sigma(k)}}\bra{x_k}$ is the permutation operator and ${\rm id}$ stands for the indentity permutation. Consequently for the associated Choi operators it holds $\mathfrak{M}_\sigma=(I\otimes P_\sigma)\mathfrak{M}_{\rm{id}}(I\otimes P_\sigma)^\dagger$.}
  Let us denote by $\T_1,\dots\T_m$ ($m=n!$) the operators forming the tester used to discriminate $\mathfrak{M}_1,\dots,\mathfrak{M}_m$, respectively. The performance is characterized by conditional probabilities $p(\sigma|\sigma^\prime)=\rm{tr}[\mathfrak{M}_\sigma \T_{\sigma^\prime}]$. Intuitively, the closer this conditional probability is to $\delta_{\sigma\sigma^\prime}$ the better, but there are various ways how to quantify the performance. \textcolor{black}{Most common one is the average error probability $p_{\rm error}=1-\frac{1}{n!}\sum_\sigma p(\sigma|\sigma)$, where $\frac{1}{n!}\sum_\sigma p(\sigma|\sigma)$ is the corresponding average success probability. Minimizing this function one finds \textit{minimum error labeling}.} Introducing additional inconclusive outcome $\T_{\rm inconclusive}$ and demanding $p(\sigma|\sigma^\prime)=0$ whenever $\sigma\neq\sigma^\prime$ is known as \textit{unambigous labeling}, for which the conclusions are error-free, but the labeling fails with probability $p_{\rm failure}=\frac{1}{n!}{\rm tr}[\sum_\sigma\mathfrak{M}_\sigma \T_{\rm inconclusive}]$. When $p_{\rm error}=p_{\rm failure}=0$ we speak about \textit{perfect labeling}.  

\section{Perfect labeling}
\noindent Perfect discrimination of two quantum channels in single use, $\mathcal{M}_1$ and $\mathcal{M}_2$, corresponds to the scenario when there is no error in each of the conclusions associated with the discrimination task. That is, the average error $p_{\text{e}} =  \mu \text{Tr}\{\mathfrak{M}_1\mathsf{T}_2\}+(1-\mu) \text{Tr}\{\mathfrak{M}_2\mathsf{T}_1\}=0$, where $\mu$ characterizes the \textit{a priori} bias between the channels. In our case there is no bias, thus both labelings are equally likely and $\mu=1/2$. The condition itself (independently of $\mu$) translates to existence of testers $\{\T_1, \T_2\}$ with normalisation \textcolor{black}{$\xi \otimes I$} satisfying \cite{chiribella2008memory,ziman2008process} the identity
\begin{equation}
\mathfrak{M}_1(\xi \otimes I) \mathfrak{M}_2 = O\,
\end{equation}
\textcolor{black}{for some normalisation $\xi$. Let us recall this condition is necessary and sufficient for the single-shot perfect distinguishability of quantum channels and is replacing the orthogonality condition $\varrho_1\varrho_2=O$ for perfect discrimination of two quantum states. In this section we will evaluate the validity of this condition for the labeling problem.}

\subsection{Binary measurements}
\noindent Let us start with the labeling of so-called binary observables characterized by pair of effects $M_1$ and $M_2=I-M_1$. There are only two permutations of outcomes determining observables $\M_{12}$ and $\M_{21}$ associated with the following Choi operators
\begin{eqnarray}
  \label{eq:binary_choi}
\mathfrak{M}_{12}=M_1^\top\otimes\dyad{1}+M_2^\top\otimes\dyad{2}\,,\\
\mathfrak{M}_{21}=M_1^\top\otimes\dyad{2}+M_2^\top\otimes\dyad{1}\,.
\end{eqnarray}
		
\begin{theorem}\label{theorem-binary}
A binary observable associated with effects $M_1, M_2$ ($M_1\neq M_2$) can be perfectly labeled in a single-shot if and only if at least one of the effect operators is rank deficient.
\end{theorem}
\begin{proof}
If $M_1=M_2$, then the observable is labeled (in fact $M_1=M_2=1/2 I$, where $I$ is identity on $d$-dimensional system). If $M_1\neq M_2$,
the perfect discrimination condition
$\mathfrak{M}_{12}(\xi\otimes I)\mathfrak{M}_{21}=O$ implies
\begin{equation}
M_1^\top\,\xi\, M_2^\top = O\,.
\end{equation}
Suppose both operators $M_1$, $M_2$ are invertible, then also their transpositions are invertible and we may apply inverse operators to obtain $\xi=O$ contradicting the assumption $\xi$ is a density operator. That is, no binary observable with full rank effects can satisfy the perfect discrimination condition. On the contrary, suppose that there exists at least one state vector $\ket{\varphi}$ such that $M_1 \ket{\varphi} = 0$, or $M_2\ket{\varphi}=0$. Then $M_1^\top\dyad{\varphi}M_2^\top=0$, thus, the perfect discrimination condition holds.
\end{proof}

This result is clear. If the measurement device is used only once at most one of the outcomes is recorded. The purity of successful probe state $\xi=\dyad{\varphi}$ implies the potential auxiliary system is uncorrelated, rendering it irrelevant for labeling. Consequently, a single click can only reveal some information about the recorded outcome. And, this information is ``negative" in the following sense. Recording an outcome implies the effect is not the one having the vector $\varphi$ in its kernel. For example, if $M_1\ket{\varphi}=0$, then the recorded outcome is not described by the effect $M_1$. In the considered binary case this implies the unrecorded outcome can be associated with $M_1$ and the recorded one with $M_2$. 
		
\subsection{Non-binary measurements}
\noindent We refer to an observable with more than two outcomes as \textit{non-binary}. The ``negative" logic used to label binary observables cannot be extended to non-binary ones. Knowing the recorded outcome is not described by an effect $M_1$ does not help in general to identify neither the recorded, nor the unrecorded outcomes. It is straightforward to see that the recording of a single outcome of the unlabeled device cannot reveal information about all the outcomes. As such, we can formulate the following theorem.

\begin{theorem}\label{theorem-nonbinary}
  Perfect labeling of a non-binary observable with effects $M_1,\dots,M_n$ (all different, i.e. $M_j\neq M_k$ for all $j\neq k$) requires at least $(n-1)$ uses of the measurement device.
\end{theorem}
\begin{proof}
  Using the device once we record exactly one outcome that can be the subject
  of labeling. In general, only recorded outcomes can be labeled, thus,
  $n$-valued observable demands at least $n-1$ uses.
\end{proof}

\noindent Let us note that by identifying $n-1$ outcomes of the considered observable
we obtain a binary observable described by effects $X_1=M_1$ and $X_2=M_2+\cdots + M_n$. This ``binarisation" is a purely abstract construction, because for the unlabeled device itself we do not know which outcomes are identified. However, if it happens the conditions of Theorem \ref{theorem-binary} are met and $X_2$
is rank deficient, then we can label the outcome of non-binary observable
associated with $M_1$. In the second use of the device we may choose a different binarisation and potentially label a different outcome. Altogether there are $n$ different binarisations of this form, but exploiting $n-1$ of them is sufficient to accomplish the goal of perfect labeling.

Let us stress that the condition that all the effects are different is important. The situation might be different if some of the outcomes are described by the same effects. In what follows we will show an example of non-binary observable for which all outcomes can be labeled in single shot. Consider an observable with effects $M_1=\dyad{\varphi}$ and $M_2=M_3=\dots=M_n=\frac{1}{n-1}(I-\dyad{\varphi})\equiv M_1^\perp$. Using the probe state $\ket{\varphi}$, the recorded outcome is necessarily associated with the effect $M_1$, but simultaneously, the remaining effects are described by the effect $M_1^\perp$. One use of the observable is sufficient to accomplish perfect labeling. It is straightforward to design similar pathological examples of $n$-valued observables that can be perfectly labeled in $m$ uses for all $m\leq n$. But for generic observables the number of outcomes determines the number of required uses to completely label all the outcomes.

It also follows the single shot setting allows for {\it partial labeling} of the recorded outcome. For example, consider an unlabeled observable with effects $M_1,\dots, M_n$ with $m$ identical effects, i.e. $M_1=M_2=\cdots=M_m$ for $m<n$ (in case $m=n$ the outcomes are labeled). Assume $\tilde{M}=M_1+\cdots + M_m=mM_1$ has eigenvalue one, then there exist a vector state $\ket{\varphi}$ such that $\tilde{M}\ket{\varphi}=\ket{\varphi}$. That is, if the state $\ket{\varphi}$ is measured, then necessarily the recorded outcome is associated with effect $M_1$ and the task of partial labeling is accomplished. Let us note that using the same state $\ket{\varphi}$ also in second use does not guarantee labeling of another unlabeled outcome.

\begin{proposition}\label{prop:partial}
An effect $E\in\{M_1,\dots,M_n\}$ can be partially labeled in a single use of the observable, if and only if there are $1\leq m\leq n$ outcomes $j$ such that $M_j=E$ and the largest eigenvalue of the effect $E$ equals $1/m$.
\end{proposition}
\begin{proof}
  Let $\ket{\varphi}$ be eigenvector associated with the eigenvalue $1/m$. Then the probability observing any of the outcomes $M_j\neq E$ equals
  $\bra{\varphi}(\sum_{j: M_j\neq E}M_j)\ket{\varphi}=\bra{\varphi}(I-mE)\ket{\varphi}=0$. In other words, only outcomes associated with $E$ can be recorded. Therefore, the observed outcome is labeled by the effect $E$.
  \end{proof}

\section{Minimum-error labeling}
\noindent In this section we will study limitations of single shot labeling for general observables. That is, we will quantify optimal success probability for labeling given the observable is binary and used only once, that is encapsulated as the following the theorem.
\begin{theorem}
	The optimal minimum error probability for labeling a binary observable associated with the effects $\{M_1, M_2\}$ is given by
	\begin{eqnarray}
		\label{eq:minerror_labeling_0}
		p_{\text{e}}&=&\frac{1}{2}(1-||M_1-M_2||_2)\,.
	\end{eqnarray}
\end{theorem}
 \begin{proof} Consider a binary observable with effects $M_1,M_2$. Following Eq.\eqref{eq:binary_choi} we denote by $\mathfrak{M}_{12}$ and $\mathfrak{M}_{21}$ the Choi operators associated with two possible labelings. Let us denote by $\mathsf{T}_{12}$ and $\mathsf{T}_{21}$ the elements of tester associated with the labelings $\{M_1,M_2\}$ and $\{M_2,M_1\}$, respectively, and satisfying the normalisation condition $\mathsf{T}_{12}+\mathsf{T}_{21}=\xi\otimes I$, where $\xi$ is a density operator. The average error probability reads
\begin{equation}
p_{\text{e}} = \frac{1}{2}\text{tr}[\mathsf{T}_{12} \mathfrak{M}_{21}]+\frac{1}{2}\text{tr}[\mathsf{T}_{21}\mathfrak{M}_{12}]\,.
\end{equation}
Using the definitions and tester's normalization we obtain
\begin{eqnarray}
  \nonumber p_{\text{e}}&=&
  \frac{1}{2}\text{tr}\left[\mathsf{T}_{12}(\mathfrak{M}_{21}-\mathfrak{M}_{12})+(\xi\otimes I)\mathfrak{M}_{12} \right] \\
  \nonumber &=& \frac{1}{2}\text{tr}[\mathsf{T}_{12}((M_1^T-M_2^T)\otimes(\dyad{1}-\dyad{2}))]+\frac{1}{2}\,.
\end{eqnarray}
In the spectral form $M_1^T-M_2^T=\sum_x \mu_x \ket{\omega_x}\bra{\omega_x}$, where $\mu_x$ are real eigenvalues and $\ket{\omega_x}$ are the associated eigenvectors. It follows the operator $(M_1^T-M_2^T)\otimes(\dyad{1}-\dyad{2})$ has $2d$ eigenvalues $\pm\mu_x$ (including multiplicities) and error probability equals
\begin{eqnarray}
  \nonumber
  p_{\text{e}}&=&
  \frac{1}{2}+\frac{1}{2}\sum_x \mu_x
  \bra{\omega_x\otimes 1} \mathsf{T}_{12} \ket{\omega_x\otimes 1}\\
  \nonumber
  & & \qquad -\frac{1}{2} \sum_x \mu_x \bra{\omega_x\otimes 2} \mathsf{T}_{12} \ket{\omega_x\otimes 2}
\end{eqnarray}
In order to minimize the error probability we need to suppress positive terms and maximize the negative one. The negativity, or positivity of the terms depends solely on the value of $\mu_x$, because the operator $\mathsf{T}_{12}$ is positive. We know there is one positive and one negative term for each nonzero $|\mu_x|$. Set $\mathsf{T}_{12}$ to vanish on all positive terms, i.e. $\mathsf{T}_{12}\ket{\omega_x\otimes j}=0$ if $(-1)^{(j-1)}\mu_x$ is positive. Then,
\begin{eqnarray}
  \nonumber
  p_{\text{e}}&=&
  \frac{1}{2}-\frac{1}{2}\sum_x |\mu_x|
  \bra{\omega_x\otimes j_x}\mathsf{T}_{12}\ket{\omega_x\otimes j_x}
  \,,
\end{eqnarray}
where $j_x=1$ if $\mu_x<0$ and $j_x=2$ if $\mu_x>0$. Due to normalisation
$\mathsf{T}_{12}\leq \xi\otimes I$, it follows
$\bra{\omega_x\otimes j}\mathsf{T}_{12}\ket{\omega_x\otimes j}
\leq\bra{\omega_x}\xi\ket{\omega_x}$. The numbers $\bra{\omega_x}\xi\ket{\omega_x}$ define a probability, thus, the values $q_k=\bra{\omega_x\otimes j_x}\mathsf{T}_{12}\ket{\omega_x\otimes j_x}$ form a subnormalized probability distribution. Our goal is to maximize the expression $\sum_x |\omega_x| q_x$ over subnormalized probability distributions. It turns out, the maximum is achieved when $\mathsf{T}_{12}=\ket{\omega_x\otimes j_x}\bra{\omega_x\otimes j_x}$ associated with the largest eigenvalue $|\mu_x|$. That is, the minimum error probability equals
\begin{eqnarray}
  \label{eq:minerror_labeling}
  p_{\text{e}}&=&
  \frac{1}{2}(1-\max_x |\mu_x|)=\frac{1}{2}(1-||M_1-M_2||_2)\,,
\end{eqnarray}
where $||\cdot||_2$ is the operator 2-norm defined as
$||X||_2=\sup_{\varphi\neq 0} \ip{X\varphi}/\ip{\varphi}=\lambda_{\max}$, where
$\lambda_{\max}$ denotes the largest singular value of $X$.

\end{proof}

\noindent Because of the normalisation $M_1+M_2=I$ the identity
$M_1-M_2=2M_1-I$ holds. It follows the eigenvalues and eigenvectors
of $M_1$, $M_2$ and $M_1-M_2$ are related. Indeed, all the operators mutually commute, because $[M_1,M_2]=0$, hence, they share the same system of eigenvectors. Let us denote by $\lambda_x$ the eigenvalues of $M_1$ and $\kappa_x$ the eigenvalues of $M_2$. Then $\lambda_x+\kappa_x=1$ and
$\mu_x=\lambda_x-\kappa_x=2\lambda_x-1$. Consequently, 
the maximal value of $|\mu_x|$ is achieved either for the
maximal, or the minimal values of $\lambda_x$. Let us denote by $\lambda_{\max}$
and $\lambda_{\min}$ the maximal and the minimal eigenvalue of $M_1$.
Let us stress that the minimal eigenvalue of $M_1$ is the maximal
one for $M_2$ and {\it vice versa}. It follows the largest of the distances
of the minimal and the maximal eigenvalue
from $1/2$ (being the value for which $\mu=0$)
determines the value of the minimum error probability.
In particular, $|\mu_x|=1$ (perfect labeling) only if either $\lambda_{\min}=0$,
or $\lambda_{\max}=1$. The binary observables with all the eigenvalues of $M_1$
close to $1/2$ are the ones for which the minimum error is the largest, thus,
the labeling is least reliable.

The optimal labeling procedure consists of preparation of the state
$\ket{\omega_x}$ being the eigenvector of $M_1$ associated with the
maximal, or minimal eigenvalue $\lambda_x$ maximizing the value
of $|\mu_x|$. The effect $M_1$ is recorded with probability
$\lambda_x=\bra{\omega_x}M_1\ket{\omega_x}$ and the effect $M_2$ with
probability $\kappa_x=\bra{\omega_x}M_2\ket{\omega_x}$.
If $\lambda_x>\kappa_x$, then we conclude the recorded outcome is $M_1$
and unrecorded one is $M_2$. In this case the probability $\kappa_x$
characterizes the error and indeed introducing $|\mu_x|=\lambda_x-\kappa_x$
into Eq.~(\ref{eq:minerror_labeling})
we obtain $p_{\text{e}}=(1-(\lambda_x-\kappa_x))/2=\kappa_x$. Similarly, if
$\lambda_x<\kappa_x$ we have $|\mu_x|=\kappa_x-\lambda_x$ and
$p_{\text{e}}=\lambda_x$, thus, the recorded outcome is labeled by
the effect $M_2$ and the unrecorded by $M_1$. In other words, in the optimal
minimum error labeling procedure we label the effect
with the largest eigenvalue by performing the unlabeled
measurement on the associated eigenvector state $\ket{\omega_x}$.

For nonbinary observables we are meeting with the same basic limitation.
Single shot can reveal the identity of unrecorded outcomes only if they
are described by an identical effect. In general, we can talk about
error probability for the question of partial labeling. That is, we use
the unlabeled measurement device once and label the recorded outcome by
one of the effects $M_1,\dots, M_n$ forming the $n$-outcome observable.
Suppose $M_x\in\{M_1,\dots, M_n\}$ is chosen to be the label for the
recorded outcome. Then
$p_{\text{e}}=\sum_{j\neq x}{\rm tr}[\omega M_j]={\rm tr}\left[\omega\sum_{j\neq x}M_j\right]={\rm tr}[\omega(I-M_x)]=1-{\rm tr}[\omega M_x]$ is the probability of error for
labeling the recorded outcome by $M_x$. It is minimized if $\tr{\omega M_x}$
is maximal. Finding $\lambda_{\max}=\max_j ||M_j||$, where $||\cdot||$ is the
operator norm that equals to the largest eigenvalue of $M_j$, we obtain the
minimum error probablity reads $p_{\text{e}}=1-\lambda_{\max}$ and
$\omega=\ket{\omega_{\max}}\bra{\omega_{\max}}$, where $\ket{\omega_{\max}}$
is the eigenvector associated with the eigenvalue $\lambda_{\max}$. Let us denote
by $M_{\max}$ the effect with eigenvalue $\lambda_{\max}$. In the optimal
partial labeling procedure the outcome recorded in the measurement
of $\ket{\omega_{\max}}$ is labeled by $M_{\max}$.

\subsection{Example: Biased coin-tossing binary observables}
\noindent This family is formed by trivial effects $M_1=qI$ and $M_2=(1-q)I$. The formula for minimum error probability implies
$p_{\rm e}=1-\lambda_{\max}$, where $\lambda_{\max}$ is the largest eigenvalue of $M_1$ and $M_2$. Assuming $1/2\leq q\leq 1$ we have $\lambda_{\max}=q$. Interestingly, the minimum error for this trivial observable can achieve any value including perfect discrimination between $O$ and $I$ when $q=1$. 
For unbiased coin-tossing observable ($q=1/2$) the minimum error equals 1/2. Nevertheless, this is not completely correct. In particular, for unbiased situation no verification is needed as both outcomes are described by the same effect and there is nothing unknown about the labels. That is, contrary to the minimum-error formula, but in accordance with the Theorem~\ref{theorem-binary}, the labeling in this case is trivially free of any error.

\section{Unambiguous labeling}
\noindent The discrimination procedure is unambiguous if so-called no-error conditions
are satisfied and inconclusive result is allowed.
For the case of labeling of binary observables this means
\begin{equation}
\text{tr}[\mathfrak{M}_{12} \T_{21}] = \text{Tr}[\mathfrak{M}_{21} \T_{21}] = 0\,,
\end{equation}
with the normalization $\T_{12}+\T_{21}+\T_?=\xi\otimes I$. The operator $\T_?$
is identified with the inconclusive result when no labeling is made.
The performance of unamiguous labeling is evaluated by the so-called failure
probability
\begin{equation}
  p_{\rm f}=\frac{1}{2}{\rm tr}[\T_?(\mathfrak{M}_{12}+\mathfrak{M}_{21})]=\frac{1}{2}{\rm tr}[\T_?]
\end{equation}
associated with the occurence of inconclusive result. The goal is to design
a tester $\T_{12},\T_{21},\T_?$ such that the failure probability is minimized.
A nontrivial solution is possible if the no-error conditions can be satisfied
for nonzero operators $\T_{12},\T_{21}$. It follows $\mathfrak{M}_{12}$, or
$\mathfrak{M}_{21}$ must have nonempty kernels. This is possible only if
at least one of the effects $M_1,M_2$ is rank deficient. However, Theorem \ref{theorem-binary} implies such observable can be perfectly labeled, thus $\T_?=O$ ($p_{\rm f}=0$) and for binary observables the unambigous labeling reduces to perfect labeling. 

For nonbinary measurements the no-error conditions also require that at least one of the effects is rank deficient. However, the single use of the unlabeled measurement device does not provide sufficient information for the unamiguous labeling of all effects. We may be able to label partially a particular effect $E$, but all the situations are covered by Theorem~\ref{prop:partial}. It follows the unambigous effect labeling is always perfect in a sense the recorded results are always conclusive.
\section{Antilabeling}
\noindent Analyzing the unambiguous labeling for non-binary measurements we can observe the following. Although the no-error conditions do not allow us to label the rank deficient effects $E$, we can still use the corresponding state $\ket{\varphi}$,
for which $E\ket{\varphi}=0$ as the probe state. It is straightforward to see that recorded outcome cannot be labeled by the effect $E$, because its probability in this case vanishes. Instead of assigning a label for the recorded outcome we may exclude all the labels $E\in\{M_1,\dots,M_n\}$, for which $M_j\ket{\varphi}=0$.

This type of discrimination tasks was named antidistinguishability \cite{HK2018} and following this we will call this process antilabeling. For binary case, it trivially coincides with labeling, however, for non-binary case the question is open. Since in this study we are focused only to single-shot protocols the antilabeling can be also achieved only partially. That is, not all effects can be excluded in single run of the experiment. The question of partial antilabeling reduces to antilabeling of some fixed effect $E\in\{M_1,\dots,M_n\}$. Rank-deficiency is necessary and sufficient for partial antilabeling of such effects. Let us stress that this is not the case for the question of partial labeling. 

Consider rank-one non-binary measurements, i.e. $M_j=q_j\ket{\varphi_j}\bra{\varphi_j}$, such that $\ket{\varphi_j}\neq\ket{\varphi_k}$ if $j\neq k$. In general, these effects cannot be partially labeled unless $q_j=1$ for some $j$. However, all of them can be partially antilabeled. It is sufficient to choose a probe state $\ket{\varphi_j^\perp}$ and the recorded outcome cannot be described by the effect $M_j=q_j\ket{\varphi_j}\bra{\varphi_j}$. It means the recorded outcome is antilabeled not to be $M_j$. If there are more $M_k$ such that $M_k\ket{\varphi_j^\perp}=0$, then all such effects are also antilabeled and we may conclude the recorded outcome is not labeled by any of those effects $M_k$.

\section{Summary}
\noindent In this work, we introduced the problem of measurement labeling being a special instance of discrimination problem. We completely characterized the various form of labeling question given the unlabeled measurement device is used only once (single-shot). In particular, we showed that perfect labeling exist for binary observables if the effects are rank deficient (Theorem~\ref{theorem-binary} ). For non-binary measurements the perfect labeling generally requires more uses of the measurement device (Theorem~\ref{theorem-nonbinary}), however, we characterize the measurements for which single shot is sufficient. The Proposition~\ref{prop:partial} speficies conditions when single shot reveals partial information about the observed outcomes, thus, perfectly labels particular outcomes. Further, we evaluated the minimum-error probability for the case of binary measurements. We found (Eq.\eqref{eq:minerror_labeling}) it is proportional to 2-norm of the difference between effects forming the observable, thus by the largest eigenvalue of the effects. Unambiguous labeling for binary measurements reduces to perfect one. The case of nonbinary measurements motivated us to introduce the concept of antilabeling that excludes labels of the recorded outcomes instead of identifying them.

\textcolor{black}{The problem of labeling is a specific question that is not only of foundational interest. Naturally, it refers to situations when the identity of labels is not available, but how one can lost this type of information? It is unlikely someone will question the description of detectors, for which clicks simply means the particle is detected. However, detectors are used in more sophisticated experiments to measure quantum properties of systems. In fact, we typically combine many devices and clicks of different detectors, or their combinations constitute different (more abstract) outcomes. The risk of forgetting the labels does not originate in the physics itself, but in its users. The labeling serves as a relatively simple tool to operationally retrieve the identity of outcomes without the necessity to access (sometimes invasively) the details of experimental setup.}
	
\section*{Acknowledgements}
\noindent Authors acknowledges the support of projects APVV-22-0570 (DEQHOST) and VEGA 2/0183/21 (DESCOM). N.S.R. acknowledges Sk Sazim and Seyed Arash Ghoreishi for discussions and providing insights into selected calculations.  The same author acknowledges Nana Siddhartha Yenamandala, and Leevi Leppäjärvi and Oskari Kerppo for discussions regarding antidiscrimination.


\end{document}